\documentclass[letterpaper,12pt]{article}
\usepackage{amsfonts,amsmath,amsthm,amssymb,mathtools,dsfont}
\usepackage{geometry,color,appendix,graphicx,float,array,hyperref}
\usepackage{caption,subcaption}
\usepackage[boxed,figure]{algorithm2e} 

\usepackage{natbib}

\providecommand{\U}[1]{\protect\rule{.1in}{.1in}}

\newcommand{\R}{\mathbb{R}}
\newcommand{\E}{\mathbb{E}\thinspace }

\newcommand{\esup}{\text{ess}\sup}

\newtheorem{theorem}{Theorem}

\newtheorem{assumption}[theorem]{Assumption}

\newtheorem{corollary}[theorem]{Corollary}

\newtheorem{definition}[theorem]{Definition}
\newtheorem{example}[theorem]{Example}

\newtheorem{lemma}[theorem]{Lemma}

\newtheorem{remark}[theorem]{Remark}

\begin{document}

\title{Predictable Forward Performance Processes:\\
The Binomial Case
\thanks{%
This work was presented at the SIAM\ conferences in Financial Mathematics in
Chicago and Austin, Second Paris-Asia Conference in Quantitative Finance in Suzhou, Workshop on Stochastic Control and Related Issues in Osaka, and seminars at Oxford University, University of Michigan, Ann Arbor, and IMA, University of Minnesota. The authors would like to
thank the participants for fruitful comments and suggestions, as well as the Oxford-Man Institute of Quantitative Finance for its support and hospitality. We also thank the two anonymous referees for constructive comments that have led to an improved version of the paper. Zhou gratefully acknowledges financial support through a start-up grant at Columbia University and through the Oxford--Nie Financial Big Data Lab.}}
\author{Bahman Angoshtari\thanks{%
Department of Applied Mathematics, University of Washington.
Email: \texttt{bahmang@uw.edu}. } \and Thaleia Zariphopoulou\thanks{
Departments of Mathematics and IROM, The University of Texas at Austin,
Austin and the Oxford-Man Institute, University of Oxford, Oxford, UK.
Email: \texttt{zariphop@math.utexas.edu.} } \and Xun Yu Zhou\thanks{
Department of IEOR, Columbia University, New York, NY 10027. Email: \texttt{%
xz2574@columbia.edu.} }}
\maketitle

\begin{abstract}
We introduce a new class of forward performance processes that are
endogenous and predictable with regards to an underlying market information
set and, furthermore, are updated at discrete times. 
We analyze
in detail a binomial model whose parameters are random and updated dynamically
as the market evolves. We show that the key step in the construction of the
associated predictable forward performance process is to solve a
single-period {\it inverse} investment problem, namely, to determine,
period-by-period and conditionally on the current market information, the
end-time utility function from a given initial-time value function. We reduce this inverse problem to solving a functional equation and establish conditions for the existence and uniqueness
of its solutions in the class of inverse marginal functions.
\vspace{1ex}

\noindent \emph{Keywords}: Portfolio selection, forward performance
processes, binomial model, inverse investment problem, functional equation,
predictability.
\end{abstract}

\section{Introduction}

The classical portfolio selection paradigm is based on three fundamental
ingredients: a given investment horizon, $[0,T]$, a performance function
(such as a utility or a risk-return trade-off), $U_{T}(\cdot )$, applied
at the {\it end} of the horizon, and a market model which yields the random
investment opportunities available over $\left[ 0,T\right) .$ This triplet is
exogenously and entirely specified at initial time, $t=0.$

Once these ingredients are chosen, one then solves for
the optimal strategy $\pi ^{\ast }(\cdot )$, and derives the value function $%
U_{0}(\cdot )$ at $t=0$ as the expectation of the terminal utility of
optimal wealth. The value function thus stipulates the best possible
performance value achievable from each and every amount of initial wealth
and, hence, it can be in turn considered as a performance criterion at $t=0$
 that is \emph{consistent} with the terminal performance criterion $U_{T}(\cdot )$.
  Here, $U_{T}(\cdot )$ is exogenous, and
$\pi ^{\ast }(\cdot )$ and
$U_{0}(\cdot )$ are endogenous. The model therefore entails a
\textit{backward} approach in time, from $U_{T}(\cdot )$ to $U_{0}(\cdot )$.
This is also in accordance with the celebrated Dynamic Programming
Principle (DPP) or, otherwise, known as Bellman's principle of optimality.

Despite its classical mathematical foundations and theoretical appeal, this
approach nonetheless has several shortcomings.
Firstly, it relies heavily on the model selection for the {\it entire} investment
horizon, which is not practical, especially if the horizon is long.
%
The second difficulty is the pre-commitment, at the initial time, to a terminal
utility. Indeed, it is clearly difficult to assess and specify the performance
function when the investment horizon is sufficiently long. Moreover, a performance
criterion naturally depends on time and state (either state of nature or state of
the agent's circumstances). It is more plausible that one
knows the utility or the resulting preferred allocations for now or the
immediate future,  and then preserves them under certain consistency
criteria (see, for example, the old note of Fischer Black, \cite{Black1988}).
Thirdly, it is very seldom the case that an optimal investment problem
``terminates'' at a single horizon $T$ or whether $T$ is a priori known when the investment activity is firstly set.

The above considerations have led to the development of the so-called \textit{%
forward performance measurement}, initially proposed by
\cite{MZ2006} and later extended by the same authors in a series of papers (see \cite{MZ2009,MZ2010, Musiela-Z-SIFIN, MZ2011}) and by
others (see, for example, \cite{ElKarouiMard2013}, and \cite{NT2015}) in continuous-time market settings. The main idea of
the forward approach is that instead of fixing, as in the classical setting,
an investment horizon, a market model and a terminal utility, one starts with an initial  performance measurement and updates it \textit{forward
in time} as the market and other underlying stochastic factors evolve. The
evolution of the forward process is dictated by a forward-in-time version of
the DPP and, thus, it ensures time-consistency across all different times.


Most of the existing results on forward performance measurement have so far focused exclusively on
continuous-time, It\^o-diffusion settings, in which both trading and
performance valuation are carried out continuously in time. It was shown in
\cite{MZ2010} that the forward process is associated with an
ill-posed infinite-dimensional stochastic partial differential equation (SPDE), the
same way that the classical value function satisfies (in Markovian models) the finite-dimensional
Hamilton-Jacobi-Bellman equation (HJB). This performance SPDE has been
subsequently studied in \cite{ElKarouiMard2013}, \cite{NZ2014}, \cite{NT2015}
and, more recently, in \cite{shkolnikovSZ2015} for asset price factors
evolving at different time scales. Despite the technical challenges that
this forward SPDE presents (ill-posedness, high or infinite dimensionality,
degeneracies, and volatility specification), the continuous-time cases are
tractable because stochastic calculus can be employed and infinitesimal
arguments can be, in turn, developed.

However, the continuous-time setting has a major drawback in that it is hard
to see how exactly the performance criterion evolves from one instant  to
the next. This evolution is lost at the infinitesimal level and hidden behind
the (generally intractable) stochastic PDE.

The aim of this paper is to introduce and study
forward investment performance processes that are \textit{discrete} in
time, while trading can be either discrete or continuous in time. We will develop
an iterative mechanism through which an investor updates/predicts her performance criterion
at the next investment period,
based on both her current performance
and her assessment
of the upcoming market dynamics in the next period.
This predictability will be present in an explicit and transparent manner.

In addition to the conceptual motivation described above, there are also
practical considerations in studying the discrete-time predictable forward
performance. Indeed, in investment practice, trading occurs at discrete
times and not continuously. More importantly, typically, performance criteria are
directly or indirectly determined by individuals, such as higher-level
managers or by clients, and not by the portfolio manager. These
\textquotedblleft performance evaluators\textquotedblright \ use information
sets that are different, both in terms of content and updating frequency,
from the ones used by the portfolio manager.
Moreover, even if trading can
occur at extremely high frequencies (hence almost close to continuous
trading), performance assessment/update takes place at a much slower pace, e.g., a
senior manager will not keep track of the performance of a portfolio or
update the performance criterion as frequently as the subordinate portfolio
manager in charge of that portfolio.

In this paper, we will consider a (possibly indefinite) series of time points, $0=t_{0},$ $t_{1}, \dots,$ $t_{n},$ $\dots$, at which the performance measurement is
evaluated and updated. The (short) period between any given two neighboring
points will be called an \emph{evaluation period}. We define our  forward performance
processes in a completely analogous way to the continuous-time counterparts.
However, we choose to work with processes that are \emph{predictable} with regards to the information at
the most recent evaluation time. We elaborate on this requirement later on.

%
%

To highlight the key ideas of predictable forward performance processes, we
start our analysis with a simple, yet still rich enough setting. The market
consists of two securities, a riskless asset and a stock whose price evolves
according to a binomial model at times $0=t_{0},t_{1},.,t_{n},\dots$, at which
the forward performance evaluation also occurs. The market model is more
general than the standard binomial tree, in that the asset returns and their
probabilities are
estimated/determined only one period ahead.
Such
a setting allows for \emph{``real-time''} dynamic updating of the underlying parameters, as the market evolves from one period
to the next.

The definition of a discrete-time predictable forward performance process (see Definition \ref{def:PFU})
dictates that in each evaluation period $\left[ t_{n},t_{n+1}\right)$,
the initial performance function $U_{n}(\cdot)$ is nothing else than the value function
of an expected utility maximization problem in this period with $U_{n+1}(\cdot)$ being
the terminal utility function.
Therefore,
in generating a predictable forward performance process,
we need to solve, in each period, an investment
problem where the value function is given and the terminal utility function
is to be found. This problem, which we term a \textit{single-period inverse investment }%
problem, then needs to be solved sequentially ``period-by-period,'' conditionally on the dynamically updated information at the
beginning of this period. It turns out that the key to solving this
problem is a linear functional equation, which relates the inverse marginal
processes at the beginning and the end of each evaluation period. We analyze this equation in detail,
and establish conditions for existence and uniqueness of the solutions in the
class of inverse marginal functions.

Once such a single-period inverse investment is solved, then starting from $[0,t_1)$ and
proceeding iteratively
\textit{forward in time}, a predictable performance process is constructed
together with the optimal allocations and their wealth processes.\footnote{%
In this paper we assume that both the updating and trading take place at the same time. As
discussed above, this does not have to be the case. However, we choose to
study this parsimonious model in order to highlight the significance of
updating the performance measurement in discrete times, without getting into too much
technicality.}


The paper is structured as follows. In Section \ref{sec:def}, we introduce
the notion of predictable forward performance processes in a general market
setting. We then formulate a binomial model with random, dynamically updated
parameters, in Section \ref{sec:Binomial}. In Section \ref{sec:Problem}, we
apply the definition of predictable forward performance processes to the
binomial model, and show that their construction reduces to solving an
inverse investment problem. In Section \ref{sec:IMP}, this inverse problem
is shown to be equivalent to solving a functional
equation. We derive sufficient existence and uniqueness conditions as well as the
explicit solution to the functional equation in Section \ref{sec:FunEq}.
Finally, we present the general construction algorithm in Section \ref{sec:Algo}, and conclude in Section \ref{sec:conclusion}. Proofs of the main results are relegated to an Appendix.

\section{Predictable forward performance processes: A general definition}

\label{sec:def}

In this section, we introduce the concept of discrete-time predictable
forward performance processes in a general market model. Starting from the next section, we will
restrict the market setting to a binomial model with random, dynamically updated,
parameters, and provide a detailed discussion on the existence and construction of
such performance processes.

The investment paradigm is cast in a probability space $(\Omega ,\mathcal{F},%
\mathbb{P})$ augmented with a filtration $(\mathcal{F}_{t})$, $t\geq 0$. We
denote by $\mathcal{X}(t,x)$ the set of all admissible wealth
processes $X_{s}$, $s\geq t$, starting with $X_{t}=x$ and such that $X_{s}$
is $\mathcal{F}_{s}$-measurable. The term \textquotedblleft
admissible\textquotedblright \  \ is for now generic and will be specified once
a specific market model is introduced in the sequel.

We call a function $U:\mathbb{R}^{+}\rightarrow \mathbb{R}^{+}$ a \textit{\
utility (or performance) function} if $U\in C^{2}(\mathbb{R}^{+})$, $%
U^{\prime }>0$, $U^{\prime \prime }<0$, and satisfies the Inada conditions, $%
\lim_{x\to 0^+}U^{\prime }\left( x\right) =\infty $ and$\, \
\lim_{x\to \infty }U^{\prime }\left( x\right) =0$.


For any $\sigma $-algebra $\mathcal{G}\subseteq \mathcal{F}$, the set of
\emph{$\mathcal{G}$-measurable utility (or performance) functions} is
defined as
\begin{equation*}
\begin{split}
\mathcal{U}(\mathcal{G})& =\left \{ U:\mathbb{R}^{+}\times \Omega
\rightarrow \mathbb{R}\, \left \vert U(x,\cdot )\text{ is }\mathcal{G}\text{
-measurable for each }x\in \mathbb{R}^{+},\right. \right. \\
& \left. \text{and }U\left( \cdot ,\omega \right) \text{ is a utility
function a.s.}\right \} .
\end{split}%
\end{equation*}%
In other words, the elements of $\mathcal{U}\left( \mathcal{G}\right) $ are
entirely known (predicted) based on $\mathcal{G}$, as they are predictable with regards to the information contained in $\mathcal{G}
$. Alternatively, we may think of $U\in \mathcal{U}\left( \mathcal{G}\right)
$ as a deterministic utility function, given the information in $\mathcal{G}.$

Next, we define the discrete predictable forward performance processes.
To ease the notation, we skip the $\omega $-argument throughout.

\begin{definition}
\label{def:PFU} Let discrete time points $0=t_{0}<t_{1}<\cdots <t_{n}<\cdots
$ be given.
A family of random functions $\{U_{0},U_{1},U_{2},\cdots \}$ is a
predictable forward performance process with respect to $\left( \mathcal{F}%
_{t}\right) $ if, for $X_{n}=X_{t_{n}}$ and $\mathcal{F}_{n}=\mathcal{F}%
_{t_{n}}$, $n=0,1,2,\dots,$ the following conditions hold:

\begin{enumerate}
\item[(i)] $U_{0}$ is a deterministic utility function and $U_{n}\in
\mathcal{U}(\mathcal{F}_{n-1})$.

\item[(ii)] For any initial wealth $x> 0$ and any admissible wealth process $%
X=\{X_n\}_{n=0}^\infty\in \mathcal{X}(0,x)$,
\begin{equation*}
U_{n-1}(X_{n-1})\geq E_{\mathbb{P}}\left[ \left. U_{n}(X_{n})\right \vert
\mathcal{F}_{n-1}\right] .
\end{equation*}

\item[(iii)] For any initial wealth $x> 0$, there exists an admissible
wealth process $X^{\ast }=\{X_n^\ast\}_{n=0}^\infty\in \mathcal{X}(0,x)$ such that
\begin{equation*}
U_{n-1}\left( X_{n-1}^{\ast }\right) =E_{\mathbb{P}}\left[ \left.
U_{n}\left( X_{n}^{\ast }\right) \right \vert \mathcal{F}_{n-1}\right] .
\end{equation*}
\end{enumerate}
\end{definition}

\bigskip

%


This definition is analogous to its continuous-time counterpart (see
\cite{MZ2009}), except condition (i).
This condition is superfluous in a continuous-time model, but fundamental in a discrete-time one.
It explicitly requires that the performance function at the \textit{next}
upcoming assessment time is \textit{entirely determined} from the
information up to the \textit{present} time (hence the name ``predictable forward").

On the other hand, as in the continuous-time case, properties (ii)-(iii) draw from Bellman's
principle of optimality, which stipulates that the processes $U_{n}(X_{n})$
and $U_{n}(X_{n}^{\ast }),$ $n=0,1,\dots,$ are, respectively, a
supermartingale and a martingale with respect to the filtration $\left(
\mathcal{F}_{n}\right) .$ Since the Bellman principle underlines time-consistency, properties (ii)-(iii) directly ensure that the investment
problem is time-consistent under the predictable forward performance
criterion.

Hence, the above performance measurement is essentially \textit{%
endogenized} by the time-consistency requirements (ii)-(iii).\footnote{Note that the predictability of risk preferences is implicitly present in the
\emph{classical} expected utility in finite horizon settings, say $\left[ 0,T\right]
,$ in which trading is continuous and a deterministic utility for a single horizon  $T$ is pre-chosen at initial time $%
t_{0}=0,$ and it is thus $\mathcal{F}_{0}$-measurable. A fundamental difference,
however, is that the terminal utility function in the classical theory is
exogenous, instead of endogenous.}


Definition \ref{def:PFU} already suggests a general scheme for constructing
predictable forward performance functions in discrete times. Indeed, starting
from an initial datum $U_{0},$ given at time $t_{0}=0$, the entire family $%
U_{1},\dots,U_{n},\dots,$ can be obtained by determining $U_{n}$ from $U_{n-1}$
iteratively, $n=1,2,\dots$, in the way described below.

Properties (ii)--(iii) dictate that, for each trading period $\left[
t_{n-1},t_{n}\right] $, we have
\begin{equation}
U_{n-1}\left( X_{n-1}^{\ast }\right) =
\underset{X_{n}\in \mathcal{X}%
(t_{n-1},X_{n-1}^{\ast })}{ \esup  }
E_{\mathbb{P}}\left[ \left.
U_{n}(X_{n})\right \vert \mathcal{F}_{n-1}\right] .  \label{eq:ForwardDPP}
\end{equation}

At instant $t_{n-1}$, since $\mathcal{F}_{n-1}$ is realized, the random
functions $U_{n-1}$ and $U_{n}$ are both deterministic and so is $%
X_{n-1}^{\ast }$. This, in turn, suggests that we should consider the
following \textquotedblleft \emph{single-period}\textquotedblright \
investment problem (conditional on $\mathcal{F}_{n-1}$):
\begin{equation}
U_{{n-1}}(x)=
\underset{X_{n}\in \mathcal{X}_{n-1,n}(x)}{ \esup  }
E_{%
\mathbb{P}}\left[ U_{n}(X_{n})\middle|
\mathcal{F}_{n-1}
\right],
\label{reverseMerton}
\end{equation}%
for $x>0,$ where, with a slight abuse of notation, we use $\mathcal{X}%
_{n-1,n}(x)$ to denote the set of admissible wealths at $t_{n}$ starting at $%
t_{n-1}$ with wealth $x$.

Therefore, if we are able to determine, for each $n=1,2,\dots,$ a performance
function $U_{n}$ $\in \mathcal{U}\left( \mathcal{F}_{n-1}\right) ,$ such
that the pair $\left( U_{n-1},U_{n}\right) $ satisfies (\ref{reverseMerton}%
), then we will have an iterative scheme to construct the \textit{entire}
predictable forward performance process, starting from $U_{0}$.

One
readily recognizes that (\ref{reverseMerton}) would be the classical
expected utility problem if the objective were to derive $U_{{n-1}}$ from $%
U_{n},$ with $U_{n}$ being a deterministic utility function. Therefore, what
we consider now is an \textit{inverse} investment problem in that
we are given its initial value function and we seek a terminal utility that
is consistent with the latter, with both of these functions being deterministic
(conditionally on $\mathcal{F}_{n-1}$).

We make the following very important observation. Definition \ref{def:PFU} of
the predictable forward criterion might at first indicate that we need to
choose the full model at $t_{0}=0,$ in that we need to completely specify
both the levels of the stock return process and the related probabilities
for \textit{all} future times $t_{1}, t_2,\dots$. As mentioned earlier,
this is a very stringent requirement in the traditional framework. However,
this is \textit{not }the case in the forward setting.

Indeed, as the analysis will show in the binomial model we analyze herein, in order to construct the predictable
forward criterion and the associated optimal portfolios and wealths, we
\textit{only }need to know at the beginning of each period, say $\left[
t_{n-1},t_{n}\right) ,$ the transition probabilities, $p_{n},$ and the
values $u_{n},d_{n}$ of the return $R_{n}.$ In other words, we only need to
specify at $t_{n-1}$ the \textit{single-step }model input $\left(
p_{n},u_{n},d_{n}\right) .$ This triplet is thus  $\mathcal{F}_{n-1}-$
measurable and, as such, it captures accurately and in ``real-time'' the
evolution of the market in $\left( 0,t_{n-1}\right] .$ There is no need to
specify at $t_{n-1}$ any model input beyond $\left(
p_{n},u_{n},d_{n}\right) .$

We also remark that, herein, we are not concerned with the specific mechanism that yields the ``real-time'' updated 
model input (e.g. $\left( p_{n},u_{n},d_{n}\right)$ for the binomial setting) at $\mathcal{F}_{t-1}.$ It
may be the outcome of a dynamic sequential learning procedure or it may be
provided exogenously from a specialist, etc. The crucial point is that
there is no requirement that it is a priori modeled for the entire optimization period.

To the best of our knowledge, such inverse discrete-time problems have not been
considered in the literature.
We start in this paper with the
binomial case in which the parameters - both the transition probabilities
and price levels - are not known a priori but are updated period-by-period as the market
moves.
As we will see, while the binomial
case is one of the simplest discrete-time market models, its analysis is sufficiently rich and its results reveal the key economic insights regarding the
predictable forward performance criteria.


\section{A binomial market model with random, dynamically updated parameters%
}

\label{sec:Binomial}

We consider a market with two traded assets, a riskless bond and a stock.
The bond is taken to be the numeraire and assumed, without loss of
generality, to offer zero interest rate.\footnote{If the  bond price follows
a predictable stochastic process, the analysis herein is valid as long as one works appropriately in discounted units.}
The stock price at times $t_{0},t_{1},\dots$,
evolves according to a binomial model that we now specify. {\ }

{Let $R_{n}$ be the total return of the stock over period $\left[
t_{n-1},t_{n}\right) $. Here, $R_n$ is a random variable with two values $%
u_n>d_n.$
We assume that $R_{n}$, $u_n$, and $d_n$, $n=1,2,\dots$%
, are all random variables in a measurable space $(\Omega ,\mathcal{F})$ augmented
with a filtration }$\left( {\mathcal{F}_{n}}\right) ,$ $n=1,2,\dots,${\ with }${%
\mathcal{F}_{n}}$ {representing the information available at $t_{n}$.
Moreover, we assume that $R_{n}$ is $\mathcal{F}_{n}$-measurable and that its values, $u_n$, and $d_n$, are $\mathcal{F}_{n-1}$-measurable. In other words,
the high and low return levels for each investment period are known at the
beginning of this period, while the realized return is known at its end.

The historical measure $\mathbb{P}\thinspace $ is a
probability measure on $(\Omega,\mathcal{F})$ and the following
standard no-arbitrage conditions are satisfied. As mentioned earlier,
the specific values of the transition probabilities, say $p_{1},p_{2,}...,$
are not a priori specified at $t_{0}=0.$ Rather, it is assumed that they are
provided at the beginning of the corresponding trading period, namely, in
the trading period $\left[ t_{n-1},t_{n}\right) ,$ $p_{n}$ is
provided at $t_{n-1}$ and as such it is $\mathcal{F}_{n-1}$-measurable. The
only standing assumption (see (\textit{ii) }below) is that these
probabilities satisfy the natural no arbitrage conditions.

\begin{assumption}
\label{assump:Historical} For all $n=1,2,\dots$:

\hspace{0.9ex}(i) $0<d_n<1<u_n$, $\mathbb{P}\thinspace $-a.s.,

(ii) The transition probabilities $p_{n}$
satisfy $0<p_{n}<1.$
\end{assumption}

The investor trades between the stock and the bond using self-financing
strategies. She starts at $t_{0}=0$ with total wealth $x>0$ and rebalances
her portfolio at times $t_{n}$, $n=1,2,\dots$. At the beginning of each
period, say $[t_{n},t_{n+1})$, she chooses the amount $\pi _{n+1}$ to be
invested in the stock (and the rest in the bond) for this period. In turn,
her wealth process, denoted by $X_{n}^{\pi },$ $n=1,2,\dots$, evolves according
to the wealth equation
\begin{equation*}
X_{n+1}^{\pi }=X_{n}^{\pi }+\pi _{n+1}(R_{n+1}-1),
\end{equation*}%
with $X_{0}=x$.

The investor is allowed to short the stock but her wealth can never become
negative; thus, $\pi _{n+1}$ must satisfy
\begin{equation}
-\frac{X_{n}^{\pi }}{u_{n+1}-1}\leq \pi _{n+1}\leq \frac{X_{n}^{\pi }}{%
1-d_{n+1}};\quad n=1,2,\dots  \label{eq:NoBankruptcy_multi}
\end{equation}

We call an investment strategy $\pi =\{ \pi _{n}\}_{n=1}^{\infty }$ \emph{%
admissible} if it is self-financing, $\pi_{n}$ is $\mathcal{F}_{n-1}$-measurable,
and \eqref{eq:NoBankruptcy_multi} is satisfied $\mathbb{P}\thinspace $%
-a.s.. A wealth process $X=\{X_{n}^{\pi }\}_{n=0}^{\infty }$ is
then admissible if the strategy $\pi $ that generates it is admissible.

We recall that $\mathcal{X}(n,x)$ is the set of admissible wealth processes $%
\{X_{m}\}_{m=n}^{\infty }$, starting with $X_{n}=x$.

We also introduce the auxiliary ``single-step'' set of admissible portfolios $%
\pi _{n+1},$ chosen at $t_{n}$ for the trading period $\left[
t_{n},t_{n+1}\right)$ and assuming wealth $x$ at $t_{n},$ by
\begin{equation*}
\mathcal{A}_{n,n+1}(x)=\left \{ \pi _{n+1}:\pi _{n+1}\text{ is $\mathcal{F}%
_{n}$-measurable, }-\frac{x}{u_{n+1}-1}\leq \pi _{n+1}\leq \frac{x}{%
1-d_{n+1}},\text{ }x>0\right \} ,
\end{equation*}%
as well as the corresponding set of admissible wealth processes
\begin{equation*}
\mathcal{X}_{n,n+1}(x)=\left \{ x+\pi _{n+1}R_{n+1}:\pi _{n+1}\in \mathcal{A}%
_{n,n+1}(x),\text{ }x>0\right \} .
\end{equation*}

\begin{remark}
	Our problem formulation and results can be readily generalized for multi-asset, complete markets. In the interest of readability of the paper, however, we opt to keep the current single-asset model.
\end{remark}

\section{Problem statement and reduction to the single-period inverse
investment problem}

\label{sec:Problem}

In this section, we consider predictable forward performance processes in the binomial model,
and show that their construction reduces to
solving a series of single-period inverse investment problems.

The investor starts with an initial utility $U_{0}$ and
updates her performance criteria at times $%
t_{1},t_{2},\dots$, with the associated performance functions $U_{1},U_{2},\dots$
satisfying Definition \ref{def:PFU}.

We now present the procedure that yields the
construction of a predictable forward performance process starting from $%
U_{0}$, and determining $U_{n}$ from $U_{n-1},$ iteratively for $n=1,2,\dots
\;$.

At $t_{0}=0$, equation (\ref{eq:ForwardDPP}) becomes
\begin{equation}
U_{0}(x)=
\underset{X_{1}\in \mathcal{X}(0,x)}{ \esup  }
E\,_{\mathbb{P}%
}\left[ U_{1}(X_{1})\Big|\mathcal{F}_{0}\right] =\underset{\pi _{1}\in
\mathcal{A}_{0,1}(x)}{\sup }E\,_{\mathbb{P}}\left[ U_{1}\Big(x+\pi
_{1}(R_{1}-1)\Big)\right] ;\quad x>0.  \label{eq:IMP_n1}
\end{equation}%
Since the market parameters $(u_1,d_1,p_{1})$ and the initial
datum $U_{0}$ are known at $t_0$, finding a deterministic $\left( \mathcal{F}_{0}%
\text{-measurable}\right) $ $U_{1}$ reduces to the single-period inverse
investment problem discussed in Section \ref{sec:def}. Let us for the moment
assume that we are able to solve this inverse problem to obtain $U_{1}$.

At $t=t_{1}$, the investor observes the realization of the stock return $%
R_{1}$ and estimates the parameters $(u_2,d_2,p_{2})$ for
the second trading period $\left[ t_{1},t_{2}\right) $. Setting $n=2$ in (%
\ref{eq:ForwardDPP}) then yields
\begin{equation}
U_{1}\left( X_{1}^{\ast }\left( x\right) \right) =
\underset{X_{2}\in
\mathcal{X}(1,X_{1}^{\ast }\left( x\right) )}{\esup}
E_{\mathbb{P}}%
\left[ \left. U_{2}(X_{2})\right \vert \mathcal{F}_{1}\right] ,
\label{eq:DPP_n2}
\end{equation}%
where $X_{1}^{\ast }\left( x\right) $ is the optimal wealth generated at $%
t_{1}$, starting at $x$ at $t_{0}=0$, from the previous period.

It follows from the classical expected utility theory
 (see also Theorem \ref{thm:U1} below) that $X_{1}^{\ast }\left(
x\right) =I_{1}(\rho _{1}U_{0}^{\prime }(x)),$ $x>0,$ where $%
I_{1}=(U_{1}^{\prime })^{-1}$ and $\rho _{1}$ is the pricing kernel over the
period $\left[ 0,t_{1}\right) $, given by
\begin{equation*}
\rho _{1}=\frac{1-d_1}{p_{1}(u_1-d_1)}\, \mathbf{1}_{\{
R_{1}=u_1\}}+\frac{u_1-1}{(1-p_{1})(u_1-d_1)}\mathbf{1}_{\{
R_{1}=d_1\}}.
\end{equation*}
The mapping $x\rightarrow X_{1}^{\ast }(x)$ is strictly increasing for each $%
x>0$ and of full range, since $I_{1}$ and $U_{0}^{\prime }$ are both
strictly decreasing functions, $\rho _{1}>0$, and the Inada conditions yield
$X_{1}^{\ast }(0)=0$ and $X_{1}^{\ast }(\infty )=\infty $.

Since $X_{1}^{\ast }(x)$ is $\mathcal{F}_{1}$-measurable and the parameters $(u_2,d_2,p_{2})$
together with $U_{1}$ are all known at $t=t_{1}$, we deduce that \eqref{eq:DPP_n2}
reduces, with a slight abuse of notation, to finding $U_{2}\left( \cdot
\right) \in \mathcal{U}\left( \mathcal{F}_{1}\right) $ such that
\begin{equation*}
U_{1}(x)=
\underset{\pi _{2}\in \mathcal{A}_{1,2}(x)}{ \esup  }
E_{%
\mathbb{P}}\left[ \left. U_{2}\left( x+\pi _{2}(R_{2}-1)\right) \right \vert
\mathcal{F}_{1}\right];\;\;x>0,
\end{equation*}%
with $U_{1}$ given. In other words, one needs to solve yet another single-period inverse
investment problem that is mathematically identical to \eqref{eq:IMP_n1}.

At $t=t_{n}$, in exactly the same manner as above, we have to solve
\begin{equation*}
U_{n}(x)=
\underset{\pi _{n+1}\in \mathcal{A}_{n,n+1}(x)}{\esup}
E_{%
\mathbb{P}}\left[ \left. U_{n+1}\left( x+\pi _{n+1}(R_{n+1}-1)\right) \right
\vert \mathcal{F}_{n}\right] ;\quad   x>0,
\end{equation*}%
thereby deriving $U_{n+1}$ from $U_{n},$ with $U_{n+1}\in \mathcal{U}\left(
\mathcal{F}_{n+1}\right)$ and with the parameters $(u_n,d_n,p_{n})$ known at $t_n$.

Thus, all the terms of a predictable forward performance process can be
obtained, starting from any arbitrary initial wealth $x>0$ and proceeding
iteratively solving a \textquotedblleft period-by-period\textquotedblright \
inverse optimization problem. Moreover, as we show in the next section, we
also concurrently derive the optimal portfolio and  wealth processes.

To summarize, the crucial step in the entire predictable forward
construction is to solve this \textit{single-period inverse investment problem%
}. We do this in the next section.

\section{The single-period inverse investment problem}

\label{sec:IMP}

We focus on the analysis of the inverse investment problem (\ref{eq:IMP_n1}%
). To ease the presentation, we introduce a simplified notation. We set $%
t_{0}=0,t_{1}=1$ and $R_{1}=R$ taking values $u$ and $d,$ $u>1$ and $0<d<1$,
with probability $0<p<1$ and $1-p,$ respectively. We recall the risk neutral
probabilities
\begin{equation*}
q=\frac{1-d}{u-d}\text{ \  \  \ and \ \ }1-q=\frac{u-1}{u-d},
\end{equation*}%
and the pricing kernel
\begin{equation}
\rho _{1}=\rho ^{u}\mathbf{1}_{\{R=u\}}+\rho ^{d}\mathbf{1}_{\{R=d\}}:=\frac{%
q}{p}\mathbf{1}_{\{R=u\}}+\frac{1-q}{1-p}\mathbf{1}_{\{R=d\}}.
\label{StochasticFactor1p}
\end{equation}

The investor starts with wealth $X_{0}=x>0,$ and invests the amount $\pi $
in the stock. Her wealth at $t=1$ is then given by the random variable $%
X=x+\pi (R-1)$. The no-bankruptcy constraint (\ref{eq:NoBankruptcy_multi})
becomes $\underline{\pi }(x)\leq \pi \leq \overline{\pi }(x)$, with
\begin{equation*}
\underline{\pi }(x)=-\frac{x}{u-1}<0\quad \text{and}\quad \overline{\pi }(x)=%
\frac{x}{1-d}>0.
\end{equation*}

We denote the set of admissible portfolios as
\begin{equation*}
\mathcal{A}(x)=\{ \pi \in \mathbb{R},\text{ and }\underline{\pi }(x)\leq \pi
\leq \overline{\pi }(x),\text{ }x>0\}.
\end{equation*}%
Given an initial utility function $U_{0}$, we then seek a deterministic performance
function $U_{1},$ such that
\begin{equation}
U_{0}(x)=\sup_{\pi \in \mathcal{A}(x)}E_{\mathbb{P}}\left[ U_{1}\left( x+\pi
(R-1)\right) \right];\;\;x>0.  \label{eq:Merton}
\end{equation}

Let $\mathcal{U}$ be the set of deterministic utility functions. We
introduce the set of \emph{inverse marginal functions }$\mathcal{I},$
\begin{equation}
\mathcal{I}:=\left \{ I\in C^{1}(\mathbb{R}^{+}):I^{\prime
}<0,\lim_{y\rightarrow \infty }I(y)=0,\lim_{y\rightarrow 0^{+}}I(y)=\infty
\right \} .  \label{I-zero}
\end{equation}%
Note that if functions $U$ and $I$ satisfy $I=(U^{\prime })^{-1}$, then $U$
is a utility function if and only if $I$ is an inverse marginal function.

Assuming for now that a utility function $U_{1}$ satisfying \eqref{eq:Merton}
exists, we consider the inverse marginal functions
\begin{equation*}
I_{0}=(U_{0}^{\prime })^{-1}\text{ \ and \  \ }%
I_{1}=(U_{1}^{\prime })^{-1}.
\end{equation*}
Our main goal in this section is to show that the inverse investment problem %
\eqref{eq:Merton} reduces to a functional equation in terms of $I_0$ and $%
I_1 $; see \eqref{eq:fnEq} below.

The following theorem is one of the main results of this paper, establishing a direct
relationship between the inverse marginals at the beginning and at the end of
the trading period $\left[ 0,1\right] ,$ when the corresponding utilities
are related by (\ref{eq:Merton}).

\begin{theorem}
\label{thm:InvMerton_FunctionEQ} Let $U_{0},U_{1}\in $ $\mathcal{U}$ satisfy
the optimization problem (\ref{eq:Merton}). Then, their inverse marginals $I_{0}$ and $I_{1}$ must
satisfy the linear functional equation
\begin{equation}
I_{1}(ay)+bI_{1}(y)=(1+b)\,I_{0}(c\,y);\quad y>0,  \label{eq:fnEq}
\end{equation}%
where
\begin{equation}  \label{eq:ab}
a=\frac{1-p}{p}\frac{q}{1-q},\text{ \ }b=\frac{1-q}{q}\quad \  \text{and }%
\quad c=\frac{1-p}{1-q}.
\end{equation}
\end{theorem}

\begin{proof}
From standard arguments, we deduce that for all $x>0$, there exists an optimizer $\pi ^{\ast }\left( x\right)$ for (\ref{eq:Merton}) satisfying the first-order condition
\begin{equation}
p(u-1)U_{1}^{\prime }(x+\pi ^{\ast }(x)(u-1))
+(1-p)
(d-1)
U_{1}^{\prime }(x+\pi
^{\ast }(x)(d-1))=0. \label{eq:Euler2}
\end{equation}
Indeed, let $f(\pi) := \E\left[U_1\left(x + \pi (R-1) \right) \right]$. By concavity of $U_1(\cdot)$, one has
	\[
		f^{\prime\prime}(\pi) = \E\left[ (R-1)^2 U_1^{\prime\prime}\left( x + \pi (R-1) \right) \right] \le 0;\quad \underline{\pi}(x)<\pi<\overline{\pi}(x).
	\]
	Furthermore,
	\[
		f^\prime(\underline{\pi}(x)) = p (u-1)  U_1^\prime(0) + (1-p) (d-1) U_1^\prime\Big( x + \underline{\pi}(x) (d-1) \Big) = +\infty
	\]
	and
	\[
		f^\prime(\overline{\pi}(x)) = p (u-1)  U_1^\prime\Big( x + \overline{\pi}(x) (u-1) \Big) + (1-p) (d-1) U_1^\prime(0) = -\infty.
	\]
	where we used the Inada condition $U_1^\prime(0)=+\infty$ and that $x + \overline{\pi}(x) (d-1) = x + \underline{\pi}(x) (u-1) = 0$ by the definition $\overline{\pi}(x)$ and $\underline{\pi}(x)$. Therefore, for any $x>0$, there exists a unique $\pi^*(x)\in(\underline{\pi}(x),\overline{\pi}(x))$ such that $f^\prime(\pi^*(x))=0$, and \eqref{eq:Euler2} follows.

On the other hand, we have from (\ref{eq:Merton}) that
\[
U_{0}(x) = p\,U_{1}(x+\pi ^{\ast }(x)(u-1))+(1-p)\,U_{1}(x+\pi ^{\ast
}(x)(d-1)).
\]
Differentiating the above equation yields
\begin{align*}
U_{0}^{\prime }(x) &= p\,U_{1}^{\prime }(x+\pi ^{\ast }(x)(u-1))+(1-p)\,U_{1}^{\prime }(x+\pi
^{\ast }(x)(d-1))\\
&\qquad+(\pi ^{\ast })'(x)\big(p(u-1)U_{1}^{\prime }(x+\pi ^{\ast }(x)(u-1))+(1-p)(d-1)U_{1}^{\prime }(x+\pi
^{\ast }(x)(d-1))\big)
\end{align*}
and, using \eqref{eq:Euler2}, one obtains
\begin{equation}\label{eq:envelop}
	U_{0}^{\prime }(x) = p\,U_{1}^{\prime }(x+\pi ^{\ast }(x)(u-1))+(1-p)\,U_{1}^{\prime }(x+\pi
^{\ast }(x)(d-1)).
\end{equation}
Solving the linear system (\ref{eq:Euler2})-\eqref{eq:envelop} gives
\begin{equation*}
U_{1}^{\prime }(x+\pi ^{\ast }(x)(u-1))=\frac{(1-d)}{p(u-d)}U_{0}^{\prime
}(x)
\end{equation*}
and
\begin{equation*}
U_{1}^{\prime }(x+\pi ^{\ast }(x)(d-1))=\frac{(u-1)}{(1-p)(u-d)}%
U_{0}^{\prime }(x).
\end{equation*}%
Therefore, the optimal allocation function $\pi^*(x)$ satisfies
\begin{equation}
\begin{cases}
x+\pi ^{\ast }(x)(u-1)&=I_{1}\left( \frac{1-d}{p(u-d)}U_{0}^{\prime
}(x)\right) , \\
x+\pi ^{\ast }(x)(d-1)&=I_{1}\left( \frac{u-1}{(1-p)(u-d)}U_{0}^{\prime
}(x)\right),%
\end{cases}
\label{Xstar}
\end{equation}%
from which we obtain the solution
\[
\pi ^{\ast }(x)=\frac{1}{u-d}\left( I_{1}\left( \frac{1-d}{p(u-d)}%
U_{0}^{\prime }(x)\right) -I_{1}\left( \frac{u-1}{(1-p)(u-d)}U_{0}^{\prime
}(x)\right) \right);\;\;x>0.
\]
Substituting the above in either of the equations in \eqref{Xstar} yields
\[
\frac{1-d}{u-d}\,I_{1}\left( \frac{1-d}{p(u-d)}U_{0}^{\prime }(x)\right) +%
\frac{u-1}{u-d}\,I_{1}\left( \frac{u-1}{(1-p)(u-d)}U_{0}^{\prime }(x)\right)
=x.
\]
Changing variables $x=I_{0}\left( \frac{(1-p)(u-d)}{u-1}y\right)$%
, $y>0,$ the above becomes
\[
I_{1}\left( \frac{(1-p)(1-d)}{p(u-1)}\,y\right) +\frac{u-1}{1-d}\,I_{1}(y)=%
\frac{u-d}{1-d}\,I_{0}\left( \frac{(1-p)(u-d)}{u-1}y\right);\;\; y>0.
\]
Noting (\ref{eq:ab}) we conclude.
\end{proof}


The next theorem shows how to recover
$U_{1}$ from $I_{1}$ and  derives
the optimal portfolio $\pi ^{\ast }\left( x\right) $ and its wealth $X^{\ast
}\left( x\right) $.

%

\begin{theorem}
\label{thm:U1} Let $U_{0}$ be a utility function and $I_{0}$ be its inverse
marginal, and $I_1$ be an inverse marginal solving the
functional equation \eqref{eq:fnEq}. Let also $\rho_1$ be the pricing kernel given by \eqref{StochasticFactor1p}. Then, the following statements hold.

\begin{enumerate}
\item[(i)] The function $U_{1}$ defined by
\begin{equation}
U_{1}(x):=U_{0}(1)+E_{\mathbb{P}}\left[ \int_{I_{1}(\rho _{1}U_{0}^{\prime
}(1))}^{x}I_{1}^{-1 }(\xi )d\xi \right];\;\; x>0,  \label{U-one}
\end{equation}%
is a well-defined utility function.

\item[(ii)] We have
\begin{equation*}
U_{0}(x)=\sup_{\pi \in \mathcal{A}(x)}E_{\mathbb{P}}\left[ U_{1}\left( x+\pi
(R-1)\right) \right];\;\; x>0.
\end{equation*}

\item[(iii)] The optimal wealth $X_{1}^{\ast }(x)$ and the associated optimal
investment allocation $\pi ^{\ast }\left( x\right) $ are given,
respectively, by
\begin{align*}
X_1^{\ast }(x)& =I_{1}(\rho _{1}U_{0}^{\prime }(x))=X^{\ast ,u}(x)\mathbf{1}
_{\{R=u\}}+X^{\ast ,d}(x)\mathbf{1}_{\{R=d\}}\\
\intertext{and}
\pi ^{\ast }\left( x\right) & =\frac{X^{\ast ,u}(x)-X^{\ast ,d}(x)}{u-d},
\end{align*}
with
\begin{equation*}
X^{\ast ,u}=I_{1}\left( \frac{q}{p}U_{0}^{\prime }(x)\right) \quad \text{
and }\quad X^{\ast ,d}=I_{1}\left( \frac{1-q}{1-p}U_{0}^{\prime }(x)\right) .
\end{equation*}
\end{enumerate}
\end{theorem}

\begin{proof}
	See Appendix \ref{app:U1}.
\end{proof}

\begin{remark}
As shown in the proof of Theorem \ref{thm:U1}, we can replace \eqref{U-one}
with
\begin{equation*}
U_{1}(x):=U_{0}(c)+E_{\mathbb{P}}\left[ \int_{I_{1}(\rho _{1}U_{0}^{\prime
}(c))}^{x}I_{1}^{-1}(\xi )d\xi \right];\;\; x>0,
\end{equation*}%
 for any arbitrary constant $c>0$. The choice of $c$ does not change
the value of $U_{1}(x)$, neither the optimal policies.
\end{remark}

As the results of Theorem \ref{thm:U1} indicate, the inverse investment problem \eqref{eq:Merton}
essentially reduces to solving the functional equation \eqref{eq:fnEq}. We study this equation next.

\section{A functional equation for inverse marginals}

\label{sec:FunEq}

In this section, we analyze the linear functional equation (\ref%
{eq:fnEq}),
with $I_{0}$ given and $I_{1}$ to be found, for positive constants $a,b,c,$
given by (\ref{eq:ab}).
We provide conditions for the existence and uniqueness of its solutions and,
in particular, solutions in the class of inverse marginal functions.

First of all, we note that the solution to (\ref%
{eq:fnEq}) is known in the literature for $b < 0$ (e.g. \cite{PM1998}). Unfortunately, in our case $b=\frac{1-q}{q}>0$ for which
we are not aware of any results to the best of our knowledge.

When $a=1$, the unique solution is trivially $%
I_{1}(y)=I_{0}(y).$ This is economically intuitive. If $%
p=q, $ then essentially there is no risk premium to exploit. As a result,
when $r=0$ as assumed herein, the
pricing kernel becomes a constant, $\rho =1,$ and the optimal wealth reduces
to $X^{\ast }\left( x\right) =x.$ In turn, the value function (at $t=0)$
coincides with the terminal utility. So
the forward performance remains constant, $%
U_{0}\left( x\right) =U_{1}\left( x\right), $ and thus their inverse
marginals $I_{0}$ and $I_{1}$ coincide.\footnote{This is also in accordance with the time-monotone forward processes in the
continuous-time setting. For example, in \cite{Musiela-Z-SIFIN},
it is shown that this forward performance is given by $U\left( x,t\right)
=u\left( x,\int_{0}^{t}\left \vert \lambda _{s}\right \vert ^{2}ds\right) ,$
with $u\left( x,t\right) $ being a deterministic function and the process $\lambda
$ is the market price of risk. If $\lambda \equiv 0,$ then $%
U\left( x,t\right) =u\left( x,0\right) =U\left( x,0\right) ,$ for all $t\ge0.$} Indeed, there is no reason to modify the performance function in a market with no investment opportunities.

\bigskip

Henceforth we assume that $a\neq 1$. We start with an example showing that a general solution of (\ref%
{eq:fnEq}) may not be unique, even if we restrict the solutions to
inverse marginal functions.

\begin{example}
\label{eq:CRRA_NonUnique}
{\rm Let $I_{0}(y)=y^{\log _{a}b}$, $%
y>0$, for constants $a,b>0$ such that $\log_a b <0$. It is easy to check that the function $I_{1}(y)=\delta y^{\log _{a}b}$%
, $y>0$, with $\delta =\frac{(1+b)}{2b\,c^{-\log _{a}b}}>0$, is a solution to \eqref{eq:fnEq}.

However, this particular solution is not the only solution. Indeed, consider any differentiable anti-periodic function,
say $\Theta (z)=-\Theta (z+\ln a),$ for which there exists a constant $M>0$
such that
\begin{equation*}
\underset{z\in \mathbb{R}}{\sup }\left( |\Theta (z)|,|\Theta ^{\prime
}(z)|\right) <M<-\delta\, \frac{\log _{a}b}{1-\log _{a}b}=-\frac{(1+b)\log
_{a}b}{2b\,c^{-\log _{a}b}(1-\log _{a}b)}.
\end{equation*}%
For instance, $\Theta (x)=M\sin (\frac{x}{\ln a}\pi )$ is such a function. One can then
directly check that the function
\begin{equation*}
\tilde{I}_{1}(y)=y^{\log _{a}b}\left( \delta +\Theta (\ln y)\right);\;\;y>0
\end{equation*}%
 is a solution.

As a matter of fact, both solutions $%
I_{1}$ and $\tilde{I}_{1}$ are inverse marginals. This is obvious for $%
I_{1}$. As for $\tilde{I}_{1},$ we have lim$_{y\to\infty }\tilde{I}%
_{1}(y)=0$ since $\log _{a}b<0$. Moreover, it follows from the inequality $\tilde{I}_{1}(y)\geq
y^{\log _{a}b}(\delta -M)$, $y>0$, that lim$_{y\to 0^+}\tilde{I}%
_{1}(y)=\infty $. Furthermore,
\begin{equation*}
\begin{split}
\tilde{I}_{1}^{\prime }(y)& =y^{\log _{a}b-1}\log _{a}b\left( \delta +\Theta
(\ln y)+\frac{\Theta ^{\prime }(\ln y)}{\log _{a}b}\right) \\
& \leq y^{\log _{a}b-1}\log _{a}b\left( \delta -\frac{M\log _{a}b-M}{\log
_{a}b}\right) <0;\;\;y>0.
\end{split}%
\end{equation*}
Thus, in general, there is no uniqueness even among inverse marginal functions.}
\end{example}

The above example suggests that we need additional conditions to ensure uniqueness.
To identify these conditions, we first note that (\ref%
{eq:fnEq}) is a
functional equation of the more general form
\begin{equation}
F\left( f(y)\right) =g(y)F(y)+h(y),  \label{F-G}
\end{equation}%
with $f$, $g$, and $h$ given functions, $y\in \mathcal{Y}\subseteq \mathbb{R}
$ and $F$ to be found. The equations of this type have been studied in the literature; see \cite{kuczmaEtAl1990} and the references therein for
a general exposition.

In general, such equations have many solutions. A trivial example is $F(y+1)=F(y),$ $%
y\in \mathbb{R}$, for which any periodic function with period 1 is a
solution. Such non-uniqueness often renders the underlying equation
inapplicable for concrete problems, where a single well-defined solution is
usually needed. For the general equation (\ref{F-G}),
conditions for the uniqueness of solutions usually limit the set of solutions by imposing additional
assumption on $F(y_{0})$, where $y_{0}$ is a fixed point for $f$: $%
f(y_{0})=y_{0}.$ In the example of the equation $F(y+1)=F(y),$ $y\in \mathbb{%
R}$, if we require a solution to be such that
$\lim_{y \to  \infty }F(y) = a \in\R$, then $F=a$
becomes the only possible solution. Note here that $%
\infty $ is actually a fixed point of the function $f(y)=y+1$.

For equation (\ref{eq:fnEq}), we have that $f\left( y\right) =ay,$ $g\left( y\right) =-b$
and $h\left( y\right) =(1+b)\,G(c\,y).$ Therefore, uniqueness conditions
should impose additional assumptions on $F$ at $y_{1}=0$ and $y_{2}=\infty $%
, which are the fixed points of $f(y)=ay$.

\bigskip

We start with the following auxiliary result in which we provide \textit{%
general} uniqueness conditions for equation (\ref{eq:fnEq}). Afterwards, we will strengthen the results for the family of inverse marginals.

\begin{lemma}
\label{prop:Uniqueness} Let $I_0$ be given. Then,
there exists at most one solution to (\ref{eq:fnEq}), say $I$, satisfying $\lim_{y \to 0^+}y^{-\log _{a}b}I(y)=0$. Similarly, there exists at most one solution satisfying $\lim_{y \to  \infty }y^{-\log _{a}b}I(y)=0$.
\end{lemma}

\begin{proof}
See Appendix \ref{app:lemma7}.
\end{proof}

We note that the function $\tilde{I}_{1}$ in Example \ref{eq:CRRA_NonUnique}
satisfies neither conditions in Lemma \ref{prop:Uniqueness}, and thus
uniqueness fails.

\bigskip


Next, we state the main result for this section, which provides sufficient
conditions for existence and uniqueness of solutions to (\ref{eq:fnEq})
that are inverse marginal functions.

\bigskip

%
%
%

\begin{theorem}
\label{thm:ExistenceOfInvMarginal} Let $I_{0}$ in (\ref{eq:fnEq}) be an
inverse marginal, i.e. $I_{0}\in \mathcal{I}$ with $\mathcal{I}$ defined
in \eqref{I-zero}. Define the functions
\begin{equation}
\Phi _{0}(y)=I_{0}(a\,c\,y)-bI_{0}(c\,y)\quad \text{ and }\quad \Psi
_{0}(y)=y^{-\log _{a}b}I_{0}(c\,y);\;\;y>0.  \label{eq:Psi_0}
\end{equation}%
The following assertions hold:

\begin{enumerate}
\item[(i)] If $\Phi _{0}$ is strictly increasing and, either $a>1$ and $%
\lim_{y \to  \infty }\Psi _{0}(y)=0$ or $a<1$ and $\lim_{y\to 0^+}\Psi _{0}(y)=0$, then a solution of (\ref{eq:fnEq}) is given
by
\begin{equation}
I_{1}(y)=\frac{1+b}{b}\sum_{m=0}^{\infty }\left( -1\right)
^{m}b^{-m}I_{0}(\,a^{m}\,c\,y);\;\;y>0.  \label{eq:fnEq_ParSol1}
\end{equation}

\item[(ii)] If $\Phi _{0}$ is strictly decreasing and, either $a>1$ and $%
\lim_{y \to 0^+}\Psi _{0}(y)=0$ or $a<1$ and $\lim_{y \to  \infty
}\Psi _{0}(y)=0$, then a solution of (\ref{eq:fnEq}) is given
by
\begin{equation}\label{eq:fnEq_ParSol2}
I_{1}(y)=(1+b)\sum_{m=0}^{\infty }(-1)^{m}b^{m}I_{0}(a^{-\left( m+1\right)
}c\,y);\;\;y>0.
\end{equation}

\item[(iii)] In parts (i) and (ii), the corresponding $I_{1}$ satisfies the
uniqueness condition(s) of Lemma \ref{prop:Uniqueness}  and, moreover, $I_{1}\in \mathcal{I}$,
i.e., $I_{1}$ preserves the inverse marginal properties.

\item[(iv)] The function $I_{1}$ in parts (i) and (ii), respectively, is the
only positive solution of (\ref{eq:fnEq}). It is also the
only inverse marginal function that solves (\ref{eq:fnEq}).
\end{enumerate}
\end{theorem}

\begin{proof}
See Appendix \ref{app:ExistenceOfInvMarginal}.
\end{proof}

Now, we apply the above results to the case when the initial utility is a power function. The following
example provides results complementary to the ones in Example \ref%
{eq:CRRA_NonUnique} where uniqueness is lacking since the conditions of
Lemma \ref{prop:Uniqueness} are not satisfied.

\begin{corollary}
\label{eq:Unique} Let $U_{0}(x)=\left( 1-\frac{1}{\theta }\right)^{-1} x^{1-%
\frac{1}{\theta }}$, $x>0$, and assume that $1\neq \theta>0$, $\theta \neq -\log _{a}b$,
with $a,b,c>0$ given by \eqref{eq:ab}. Then, the following assertions hold:

\begin{enumerate}
\item[(i)] The unique inverse marginal function that satisfies the functional
equation \eqref{eq:fnEq} with initial $I_{0}(y)=y^{-\theta }$ is
given by
\begin{equation}
I_{1}(y)=\delta y^{-\theta };\;\;y>0, \label{eq:CRRA_InvMarg_sol}
\end{equation}
where $\delta=\frac{1+b}{%
c^{\theta }(a^{-\theta }+b)}.$
\item[(ii)] The unique utility function $U_{1}$ that satisfies the inverse
investment problem (\ref{eq:Merton}) is given by
\begin{equation*}
U_{1}(x)=\delta ^{\frac{1}{\theta }}\left( 1-\frac{1}{\theta }\right)^{-1} x^{1-%
\frac{1}{\theta }}=\delta ^{\frac{1}{\theta }}U_{0}\left( x\right);\;\;x>0.
\end{equation*}%
\item[(iii)] The corresponding optimal allocation is given by
\begin{equation*}
\pi ^{\ast }(x)=\frac{\delta (p/q)^{\theta}-1}{u-1}\; x;\;\;x>0.
\end{equation*}
\end{enumerate}
\end{corollary}

Therefore, if we start with an initial power utility $U_{0}$, then the forward utility  at $t=1$ is a multiple of
the initial datum, with the constant given by $\delta ^{\frac{1}{\theta }}$. Note that $\delta$ incorporate both the preference parameter $\theta$ and also the market parameters $a$, $b$, and $c$ at the \emph{beginning} of the trading period $t=0$. Proceeding iteratively, the utilities for all future periods remain power functions. In other words, in the binomial setting, the (predictable) power utility preferences are preserved throughout.


We conclude this section summarizing the findings for existence and uniqueness of solutions.
If equation (\ref{eq:fnEq}) does not admit a solution, then it follows from  Theorem \ref{thm:InvMerton_FunctionEQ} that  there will be no utility function
$U_1$ satisfying equation (\ref{eq:Merton}). Hence there will be no predictable forward performance process starting from the initial marginal utility function $I_0$.
On the other hand,  (\ref{eq:fnEq}) may have more than one solutions and, in particular, more than one solutions that are inverse marginal functions. In this case, problem (\ref{eq:Merton}) has multiple solutions as well.
An open question is which of these solutions can be chosen to be the ``correct" forward utility. Lemma \ref{prop:Uniqueness}  suggests that uniqueness follows from imposing certain decay conditions for large or small wealth; this is in accordance with the well-known elasticity condition in the classical setting.

\section{Construction of the predictable forward performance process}\label{sec:Algo}

We are now ready to present the general algorithm for the construction of
forward performance processes as well as the associated optimal
investment strategies and their wealth processes. We stress that one of the
main strengths of our approach is that for every given trading period, say $[t_{n},t_{n+1}] ,$ we do not have to update the model parameters $\left(
u_{n+1},d_{n+1},p_{n+1}\right) $ for this period until time $t_{n}$ arrives. Thus, we take full advantage of the incoming information
up to time $t_{n}.$ This is in contrast with the classical setting where, as we mentioned earlier, these parameters have to be pre-specified at initial time.

The algorithm is based on repeatedly applying, conditionally on the new ``real-time'' information, the following result on the single-period
inverse investment problem \eqref{eq:Merton}.
\begin{theorem}
\label{thm:IMP} For the inverse investment problem \eqref{eq:Merton},
assume that the initial inverse marginal $I_{0}=(U_{0}^{\prime })^{-1}$
satisfies condition (i) (resp. condition (ii)) in Theorem 8, and define $%
I_{1}$ by \eqref{eq:fnEq_ParSol1} (resp. \eqref{eq:fnEq_ParSol2}).
Then, the unique solution to \eqref{eq:Merton} is given by
\[
U_{1}(x)=U_{0}(1)+E_{\mathbb{P}}\left[ \int_{I_{1}(\rho _{1}U_{0}^{\prime
}(1))}^{x}I_{1}^{-1}(\xi )d\xi \right] ;\quad x>0,
\label{U-1}
\]
where $\rho _{1}$ is as in \eqref{StochasticFactor1p}.
Moreover, the optimal wealth $X_{1}^{\ast }(x)$ and the associated optimal investment
allocation $\pi ^{\ast }\left( x\right) $ are given, respectively, by%
\[
X_{1}^{\ast }(x)=I_{1}(\rho _{1}U_{0}^{\prime }(x))=X_{1}^{\ast ,u}(x)%
\mathbf{1}_{\{R_{1}=u\}}+X_{1}^{\ast ,d}(x)\mathbf{1}_{\{R_{1}=d\}}
\]
and
\[
\pi ^{\ast }\left( x\right) =\frac{X_{1}^{\ast ,u}(x)-X_{1}^{\ast ,d}(x)}{u-d%
},
\]
where
\[
X_{1}^{\ast ,u}(x):=I_{1}\left( \frac{q}{p}U_{0}^{\prime }(x)\right) \quad \  \
\text{and }\quad X_{1}^{\ast ,d}(x):=I_{1}\left( \frac{1-q}{1-p}U_{0}^{\prime
}(x)\right) .  \label{X-1-values}
\]
\end{theorem}

\begin{proof}
The results follow directly from Theorem \ref{thm:ExistenceOfInvMarginal} and Theorem \ref{thm:U1}.
\end{proof}

%
%
%
%
%

Given an initial performance function $U_{0}$ and initial
wealth $X_{0}$, the following algorithm provides the predictable forward performance process
$\{U_1,U_2,\cdots\}$ along with the associated
optimal portfolio process $\{\pi^*_1,\pi^*_2,\cdots\}$ and the wealth process
$\{X^*_1,X^*_2,\cdots\}$ in the binomial market model.
\begin{itemize}
\item At $t=0:$ Assess the market parameters $\left(
u_1,d_1,p_{1}\right) $ for the first investment period, $\left[
0,t_{1}\right) .$ Compute
\end{itemize}

\[
q_{1}=\frac{1-d_1}{u_1-d_1},\quad a_{1}=\frac{%
q_{1}(1-p_{1})}{p_{1}(1-q_{1})},\quad b_{1}=\frac{1-q_{1}}{q_{1}},\quad
\text{and \ }c_{1}=\frac{1-p_{1}}{1-q_{1}},
\]
and
\[
\rho _{1}^{u}=\frac{q_{1}}{p_{1}}\text{\textbf{\  \  \  \  \ }and\textbf{\  \  \ }}%
\rho _{1}^{d}=\frac{1-q_{1}}{1-p_{1}}.
\]

Using $\left( a_{1},b_{1},c_{1}\right)$, check the conditions in part (i)
(resp. (ii)) of Theorem 8, and obtain the inverse marginal function $I_{1}$ from \eqref{eq:fnEq_ParSol1} (resp. %
\eqref{eq:fnEq_ParSol2}).
Then, apply Theorem \ref{thm:IMP} to compute
\[
U_{1}(x) =U_{0}(1)+p_{1}\int_{I_{1}(\rho _{1}^{u}U_{0}^{\prime
}(1))}^{x}I_{1}^{-1}(\xi )d\xi
+(1-p_{1})\int_{I_{1}(\rho _{1}^{d}U_{0}^{\prime
}(1))}^{x}I_{1}^{-1}(\xi )d\xi;\;\;x>0,
\]
\[
\pi^{\ast }_1=\frac{X_{1}^{\ast ,u}(X_0)-X_{1}^{\ast ,d}(X_0)}{u-d%
},  \label{pi-1}
\]
and
\[
X_{1}^{\ast}=X_{0}+\pi _{1}^{\ast
}\left( R_{1}-1\right),
\]
where
\[
X_{1}^{\ast ,u}(x)=I_{1}\left( \frac{q_1}{p_1}U_{0}^{\prime }(x)\right) \quad \  \
\text{and }\quad X_{1}^{\ast ,d}(x)=I_{1}\left( \frac{1-q_1}{1-p_1}U_{0}^{\prime
}(x)\right);\;\;x>0.
\]

\begin{itemize}
\item \noindent At $t=t_{n}$ ($n=1,2,\cdots$): We have already obtained $\{U_1,\cdots,U_n;I_1,\cdots,I_n\}$,
$\{\pi^*_1,\cdots,\pi^*_n\}$ and
$\{X^*_1,\cdots,X^*_n\}$.

Estimate the market parameters $%
(u_{n+1},d_{n+1},p_{n+1})$ for the upcoming investment period $\left[
t_{n},t_{n+1}\right) $. Let
\begin{align*}
	&q_{n+1}=\frac{1-d_{n+1}}{u_{n+1}-d_{n+1}},\quad a_{n+1}=\frac{%
	q_{n+1}(1-p_{n+1})}{p_{n+1}(1-q_{n+1})},\quad b_{n+1}=\frac{1-q_{n+1}}{%
	q_{n+1}},\\
	&c_{n+1}=\frac{1-p_{n+1}}{1-q_{n+1}}, \quad
	\rho _{n+1}^{u}=\frac{q_{n+1}}{p_{n+1}}\quad \text{ and }\quad \rho
	_{n+1}^{d}=\frac{1-q_{n+1}}{1-p_{n+1}}.	
\end{align*}
Check the conditions in part (i) (resp. (ii)) in Theorem \ref{thm:ExistenceOfInvMarginal}, using $\left(
a_{n+1},b_{n+1},c_{n+1}\right) $ (instead of $\left( a,b,c\right) $) and $%
I_{n}$ instead of $I_{0},$ and obtain $I_{n+1}$ from \eqref{eq:fnEq_ParSol1} (resp. %
\eqref{eq:fnEq_ParSol2}).\footnote{%
If both conditions in part (i) and (ii) do not hold, then the functional
equation \eqref{eq:fnEq} may not have a solution, or the solution may not be
unique. For the case of initial power utility $U_{0}(x)=\frac{x^{1-1/\theta }}{%
1-1/\theta }$, $\theta >0$, Example \ref{eq:CRRA_NonUnique} and Corollary %
\ref{eq:Unique} show that both condition fail at $t_{n}$ if and only if $%
\theta =-\log _{a}b>0$, in which case the solution exists but is not unique.
This case is pathological, but to solve it remains a technically interesting
question.}

Compute
\begin{equation}\label{U-n}
\begin{split}
U_{n+1}(x)& =U_{n}(1)+p_{n+1}\int_{I_{n+1}(\rho _{n+1}^{u}U_{n}^{\prime
}(1))}^{x}I_{n+1}^{-1}(\xi )d\xi \\
& \hphantom{=U_n(1)}+(1-p_{n+1})\int_{I_{n+1}(\rho _{n+1}^{d}U_{n}^{\prime
}(1))}^{x}I_{n+1}^{-1}(\xi )d\xi ;\quad x>0,
\end{split}
\end{equation}
\[
\pi _{n+1}^{\ast }=\frac{X_{n+1}^{\ast ,u}\left( X_{n}^{\ast }\right)
-X_{n+1}^{\ast ,d}\left( X_{n}^{\ast }\right) }{R_{n+1}^{u}-R_{n+1}^{d}},
\]
and
\[
X_{n+1}^{\ast }
=X_{n}^{\ast }+\pi _{n+1}^{\ast }\left( R_{n+1}-1\right) =X_{0}+\sum_{i=1}^{n+1}\pi _{i}^{\ast }\left( R_{i}-1\right),
\]
where,
\[
X_{n+1}^{\ast ,u}\left( x\right) =I_{n+1}\left( \frac{q_{n+1}}{p_{n+1}}%
U_{n}^{\prime }(x)\right) \text{ \  \  \  \ and \  \  \ }X_{n+1}^{\ast
,d}\left( x\right) =I_{n+1}\left( \frac{1-q_{n+1}}{1-p_{n+1}}U_{n}^{\prime
}(x)\right);\quad x>0.
\]
\end{itemize}

\bigskip

In summary, starting with an initial datum $U_{0},$ we have constructed for
(the end of) each trading period, say $\left( t_{n},t_{n+1}\right] ,$ $%
n=1,2,...,$ a performance criterion $U_{n+1}$ at $t_{n+1}$ that is indeed $%
\mathcal{F}_{n}-$measurable. This measurability is inherited by the same
measurability of the inverse marginal $I_{n+1}$ that enters in the lower
part of the integration in (\ref{U-n}).
Moreover, as expected, the optimal wealth $X_{n+1}^{\ast }$ is $\mathcal{F%
}_{n+1}-$measurable, given that the pricing kernel $\rho _{n+1}$ is  $%
\mathcal{F}_{n+1}-$measurable. The optimal portfolio $\pi _{n+1}^{\ast }$ is
$\mathcal{F}_{n}-$measurable, chosen at the beginning of the period $\left[t_{n},t_{n+1}\right).$


\section{Conclusions}

\label{sec:conclusion}

We have introduced a discrete time analogue of the continuous time forward performance processes, focusing on the predictability of such criteria.
%
%
Specifically, at the beginning of each evaluation period, the investor
assesses the market parameters only for this period (during which trading
may take place once or many times, in both discrete or continuous fashion). Then, she solves an
inverse single-period inverse investment model which yields the utility at the end
of the period, given the one at the beginning. The martingality and
supermartingality requirements of the forward performance process ensure that this construction, ``period-by-period
forward in time'' and adapted to the new market information, yields
time-consistent policies.

We have implemented this new approach in a binomial model with random, dynamically updated
parameters, including both the probabilities and the levels of the stock
returns.
We have then
discussed in detail how the construction of predictable forward performance
processes essentially reduces to a single-period inverse investment problem.
We have, in turn, shown that the latter is equivalent to solving a
functional equation involving the inverse marginal functions at the
beginning and the end of the trading period, and have established conditions for
the existence and uniqueness of solutions in the class of inverse marginal functions.

We have finally provided an explicit algorithm that yields the forward
performance process as well as their optimal portfolio and the associated optimal
wealth processes.

There are a number of possible future research directions. Firstly, one may depart from
the binomial model to study general discrete-time
models, while allowing for trading to be discrete or continuous. Such models are
inherently incomplete and additional difficulties are expected to arise with
regards to the derivation of the functional equation for the inverse
marginals as well as the existence and uniqueness of its solutions among suitable classes of functions.


A second direction is to enrich the predictable framework by incorporating model ambiguity.
This will allow for the specification of all possible market models only one
evaluation period ahead, thus offering substantial flexibility to narrow down the most
realistic models period-by-period as the market evolves.

From the theoretical point of view, an interesting question is to investigate
whether discrete predictable forward performance processes converge to their continuous-time
counterparts. While this is naturally and intuitively expected,
conditions on the appropriate convergence scaling need to be imposed, which
might be quite challenging due to the ill-posedness of the problem.
Such
results may also shed light to deeper questions on the construction of continuous-time
forward performance criteria related to the appropriate choice of their
volatility, finite-dimensional approximations, Markovian or path-dependent cases, among
others.

\appendix

\section{Proof of Theorem \protect \ref{thm:U1}}

\label{app:U1}

We start with the following auxiliary result, showing that the expected utility problem \eqref{eq:Merton} is
equivalent to
\begin{equation}
U_{0}\left( I_{0}(y)\right) =E_{\mathbb{P}}\left( U_{1}\left( I_{1}(\rho
_{1}\;y)\right) \right) ;\quad y>0.  \label{eq:Merton_Inv}
\end{equation}

\begin{lemma}
\label{lem:Aux} Suppose that $U_0,U_1\in \mathcal{U}$ and let $I_0$ and $I_1$
be respectively their inverse marginals. Then, \eqref{eq:Merton} holds if and only if %
\eqref{eq:Merton_Inv} holds.
\end{lemma}

\begin{proof}
	We first show that \eqref{eq:Merton} implies \eqref{eq:Merton_Inv}. Indeed, standard results in expected utility maximization yield that \eqref{eq:Merton} implies
	\[
		U_0(x) = E_{\mathbb{P}}\Big[U_1\Big(I_1\big(\rho_1 U_0^\prime(x)\big)\Big)\Big];\quad x>0,
	\]
	and \eqref{eq:Merton_Inv} is then obtained by the change of variable $y=U_0^\prime(x)$.
	
	Next, we show that \eqref{eq:Merton_Inv} yields \eqref{eq:Merton}. Define the value function $\tilde{U}$ by
	\[
		\tilde{U}(x):=\sup_{\mathcal{A}(x)}E_{\mathbb{P}}\left[ U_{1}\left( X\right)\right];\quad x>0.
	\]
	We claim that $\tilde{U}\equiv U_0$. Let $\tilde{I}$ be the inverse marginal of $\tilde{U}$. By (i), one must then have
	\[
		\tilde{U}\big(\tilde{I}(y)\big) = E_{\mathbb{P}}\Big[U_1\big(I_1(\rho_1\; y)\big)\Big];\quad y>0,
	\]
	and it follows that $\tilde{U}\big(\tilde{I}(y)\big) = U_0\big(I_0(y)\big)$, for $y>0$.
	
	Differentiating with respect to $y$ yields $\tilde{I}^\prime \equiv I_0^\prime$. Therefore $\tilde{I}(y) = I_0(y) + C$, $y>0$, for some constant $C$. Taking the limit as $y\to\infty$ and using the Inada condition $\tilde{I}(\infty) = I_0(\infty) = 0$, we deduce that  $C=0$. Therefore, we obtain $\tilde{I}\equiv I_0$, which implies $\tilde{U}^\prime(x) = U_0^\prime(x)$, for all $x>0$. Finally, we obtain
	\[
		\tilde{U}(x)=E_{\mathbb{P}}\Big[ U_{1}\big( I_{1}(\rho \tilde{U}^{\prime }(x))\big) \Big] =E_{\mathbb{P}}\Big[
		U_{1}\big(I_{1}(\rho U_{0}^{\prime }(x))\big) \Big] =U_{0}(x);\quad x>0.\qedhere
	\]
\end{proof}

\begin{proof}[Proof of Theorem \ref{thm:U1}]
	(i): From \eqref{U-one} it follows that
\[
	U_{1}(x):= U_{0}(1) + p \int_{x_u(1)}^{x} I_{1}^{-1}(\xi )d\xi + (1-p) \int_{x_d(1)}^{x} I_{1}^{-1}(\xi )d\xi;\quad x>0,
\]
where $x_u(\cdot)$ and $x_d(\cdot)$ are given by
\begin{equation}\label{eq:XuXd}
	x_i(c) = I_1\big(\rho^i\;U_0^\prime(c)\big);\quad c>0,\; i=u,d.
\end{equation}
Thus,
\[
	U_{1}^\prime(x) = p\;I_{1}^{-1}(x) + (1-p)\;I_{1}^{-1}(x) = I_{1}^{-1}(x);\quad x>0.
\]
It then follows that $I_1$ is the inverse marginal of $U_1$ and that $U_1$ is a utility function.\\

\noindent (ii): Define the function $F$ by
	\begin{equation}\label{eq:F}
		F(x,c) := U_0(c) + p \int_{x_u(c)}^x I_{1}^{-1}(\xi) d\xi + (1-p) \int_{x_d(c)}^x I_{1}^{-1}(\xi) d\xi;\quad (x,c)\in \R^+\times \R^+,
	\end{equation}
	with $x_u(c)$ and $x_d(c)$ as in \eqref{eq:XuXd}.	We claim that
	\[
		\frac{\partial F}{\partial c}(x,c) = 0;\quad   x,c>0.
	\]
	Indeed, differentiating \eqref{eq:F} with respect to $c$ and then using that $I_1^{-1}\big(x_i(c)\big) = \rho^i\;U_0^\prime(c)$, for $c>0$, we have
	\[
	\begin{split}
		\frac{\partial F}{\partial c}(x,c) &= U_0^\prime(c) - p x_u^\prime(c) G\big(x_u(c)\big) - (1-p) x_d^\prime(c) G\big(x_d(c)\big)\\
		&= U_0^\prime(c) - p x_u^\prime(c) \rho^u\;U_0^\prime(c) - (1-p) x_d^\prime(c) \rho^d\;U_0^\prime(c)\\
		&= U_0^\prime(c) \Big ( 1 - p \rho^u x_u^\prime(c) - (1-p) \rho^d x_d^\prime(c)\Big ) = 0.
	\end{split}
	\]
	To obtain the last equation, note that equation \eqref{eq:fnEq} is equivalent to
	\[
		I_0(y) = p \rho^u I_1(y\; \rho_u) + (1-p) \rho^d I_1(y\; \rho_d);\quad y>0.
	\]
	Therefore, substituting $y=U_0(c)$ and differentiating with respect to $c$ yield
	\[
	\begin{split}
		1 &= \frac{d}{d c}\left(I_0^\prime \Big(U_0^\prime(c)\Big)\right) = \frac{d}{d c}\bigg(p \rho^u I_1\Big(\rho_u\;U_0^\prime(c)\Big) + (1-p) \rho^d I_1\Big(\rho_d\;U_0^\prime(c)\Big)\bigg)\\
		  &=  p (\rho^u)^2 I_1^\prime \Big(\rho_u\;U_0^\prime(c)\Big)U_0^{\prime \prime}(c) + p (\rho^d)^2 I_1^\prime \Big(\rho_d\;U_0^\prime(c)\Big)U_0^{\prime \prime}(c)\\
		  &= p \rho^u x_u^\prime(c) + (1-p) \rho^d x_d^\prime(c).
	\end{split}
	\]
	
	Note that, by definition, $U_1(x) = F(x,1)$. Since we have showed that $\frac{\partial F}{\partial c} \equiv 0$, we must have $U_1(x) = F(x,c)$, for all $x>0$ and $c>0$. In other words, for all $x,c\in \R^+$, $U_1$ satisfies
	\[
		U_1(x) = U_0(c) + p \int_{x_u(c)}^x I_{1}^{-1}(\xi) d\xi + (1-p) \int_{x_d(c)}^x I_{1}^{-1}(\xi) d\xi.
	\]
	On the other hand, as it was shown in (i), $U^\prime_1\equiv I_{1}^{-1}$. Therefore, for all $x>0$ and $c>0$,
	\[
		U_1(x) = U_0(c) + p \Big(U_1(x) - U_1\big(x_u(c)\big)\Big) + (1-p) \Big(U_1(x) - U_1\big(x_d(c)\big)\Big),
	\]
	which, in turn, yields that
	\[
		U_0(c) = p U_1\big(x_u(c)\big) + (1-p) U_1\big(x_d(c)\big) =  E_{\mathbb{P}}\Big[U_1\big(I_1(\rho_1\; U^\prime_0(c))\big)\Big];\quad   c>0.
	\]
	This is equivalent to \eqref{eq:Merton_Inv}. Hence, (ii) follows from Lemma \ref{lem:Aux}.
	
	\noindent (iii): This part follows easily from existing results in the classical expected utility problems, if we view (\ref{eq:Merton}) as a terminal expected utility problem with $U_{1}$ now given and $U_{0}$ being its value function.
\end{proof}

\section{Proof of Lemma \ref{prop:Uniqueness}}\label{app:lemma7}

Let $F_{1}$ and $F_{2}$ be two solutions of (\ref{eq:fnEq}) that both
satisfy either conditions given in the lemma. We show that their difference $%
w:=F_{1}-F_{2}\equiv 0$.

The function $w$ satisfies the homogenous equation $w(ay)=-bw(y),$ $y>0.$
Therefore, for $k=1,2,\dots$,
\begin{equation*}
w(y)=\frac{w(ay)}{-b}=\frac{w(a^{2}y)}{(-b)^{2}}=\dots =\frac{w(a^{k}y)}{(-b)^{k}},
\end{equation*}
and
\begin{equation*}
w(y)=-bw\left( \frac{y}{a}\right) =(-b)^{2}w\left( \frac{y}{a^{2}}\right)
=\dots =(-b)^{k}w\left( \frac{y}{a^{k}}\right) .
\end{equation*}%
It then follows that for $k=\pm 1,\pm 2,\dots$ and $y>0$,%
\[
\begin{split}
	|w(y)| &= b^{k}\left \vert w\left( \frac{y}{a^{k}}\right) \right \vert =y^{\log
	_{a}b}\left( \frac{y}{a^{k}}\right) ^{-\log _{a}b}\left \vert w\left( \frac{y%
	}{a^{k}}\right) \right \vert \\
	&\leq y^{\log _{a}b}\left( \frac{y}{a^{k}}\right) ^{-\log _{a}b}\left(\left \vert
	F_{1}\left( \frac{y}{a^{k}}\right) \right \vert +\left \vert \ F_{2}\left(
	\frac{y}{a^{k}}\right) \right \vert \right).		
\end{split}
\]
The right side vanishes as either $k \to  \infty $ or $k\to
-\infty $, and we conclude.

\section{Proof of Theorem \protect \ref{thm:ExistenceOfInvMarginal}}

\label{app:ExistenceOfInvMarginal} We only show part (i) and the
corresponding statements in parts (iii) and (iv), since (ii) follows from
similar arguments.

(i) Direct substitution shows that if the infinite series in %
\eqref{eq:fnEq_ParSol1} converges, then $I_1$ satisfies equation (\ref{eq:fnEq}).
Thus, to show (i), it only remains to establish that the series converges. Note that \eqref{eq:fnEq_ParSol1} can be
written, for $y>0,$ as
\begin{equation}
I_1\left( y\right) =\frac{b}{1+b}y^{\log _{a}b}\sum_{m=0}^{\infty
}(-1)^{m}\Psi _{0}(a^{m}\,y),  \label{eq:F_Psi0}
\end{equation}%
which, by the Leibniz test for alternating series, converges if $%
\lim_{m\rightarrow \infty }\Psi _{0}(a^{m}\,y)=0$ monotonically. The fact that $\lim_{m\rightarrow \infty }\Psi _{0}(a^{m}\,y)=0$ follows directly from
either of the conditions in (i) on $a$ and $\Psi _{0}$. To show that the
convergence is monotonic, note that \eqref{eq:Psi_0} yields
\begin{equation}
\Psi _{0}(a^{m+1}\,y)-\Psi _{0}(a^{m}\,y)=b^{-m-1}y^{-\log _{a}b}\Phi
_{0}(a^{m}\,y);\quad y>0,\; \;m=0,1,\dots.  \label{eq:Psi0_Phi0}
\end{equation}%
On the other hand, since $\Phi _{0}$ is increasing and $\lim_{y \to
\infty }\Phi _{0}(y)=\lim_{y \to  \infty }\big(I_0(a\,c\,y)-b\;I_0(c\,y)\big)=0$, by Inada's condition, we must have $\Phi _{0}(y)<0$, for $y>0$. Thus, by %
\eqref{eq:Psi0_Phi0}, we deduce that $\Psi _{0}(a^{m}\,y)>\Psi _{0}(a^{m+1}\,y)$ and $%
\lim_{m \to  \infty }\Psi _{0}(a^{m}\,y)=0$ monotonically.

(iii) First, we prove that $I_1$ is strictly decreasing. Indeed, %
\eqref{eq:F_Psi0} and \eqref{eq:Psi0_Phi0} yield
\begin{equation*}
I_1\left( y\right) = \frac{b}{1+b}y^{\log_{a}b}\sum_{m=0}^\infty \Big(\Psi
_{0}(a^{2 m}\,y) - \Psi _{0}(a^{2m+1}\,y)\Big) = -\frac{1}{1+b}%
\sum_{m=0}^\infty b^{-2 m} \Phi_0(a^{2m}\, y).
\end{equation*}
It then follows that, for $y<y^{\prime }$,
\begin{equation*}
I_1(y^\prime)-I_1(y) = \frac{1}{1+b}\sum_{m=0}^\infty b^{-2 m} \Big(%
\Phi_0(a^{2m}\, y) - \Phi_0(a^{2m}\, y^\prime) \Big) < 0,
\end{equation*}
where the inequality holds because $\Phi_0$ is strictly increasing.

Using equation \eqref{eq:fnEq}, that $a,b,c>0$, that $\lim_{y \to  \infty
}I_0(y) =0$, and the monotonicity of $I_1$, we deduce that $\lim_{y \to
\infty }I_1(y) =0,$ and, hence, $I_1(y) >0$, for $y>0.$ Similarly, the fact that $%
\lim_{y \to 0^+}I_0(y) =\infty$ yields that $\lim_{y \to 0^+}I_1(y) =\infty $%
. Thus, we have shown that $I_1\in \mathcal{I}$.

Finally, the conditions in Lemma \ref{prop:Uniqueness} follow from $\Psi_0(y)\to
0$, as either $y \to 0^+$ or $y \to  \infty$, and from the inequalities
\begin{equation*}
0< y^{\log_a b} I_1(y) = \frac{I_1(y)}{I_0(c\,y)} \Psi_0(y) < \frac{b+1}{b}
\Psi_0(y);\quad   y>0,
\end{equation*}
where we used \eqref{eq:fnEq} and that $I_1(y)>0$ to obtain
\begin{equation*}
\frac{I_1(y)}{I_0(c\,y)} = \frac{(1+b) I_1(y)}{I_1(a\,y)+b\;I_1(y)} < \frac{1+b}{b}.
\end{equation*}

(iv) Repeating the last part of the arguments in part (iii) for any solution $%
\tilde{I}>0$ yields that $\tilde{I}$ satisfies the same uniqueness
condition for \eqref{eq:fnEq} as $I_1$. The result then follows directly from Lemma %
\ref{prop:Uniqueness}.

\section{Proof of Corollary \protect \ref{eq:Unique}}

\label{app:Unique}

Assertion (ii) follows from (i) and Theorem \ref{thm:U1}.
Also, one
can easily check that $I_{1}$ given by \eqref{eq:CRRA_InvMarg_sol} is thus an
inverse marginal satisfying equation \eqref{eq:fnEq}.

It only remains to show the uniqueness of solutions that are inverse
marginals. To this end, it suffices to check that the conditions of Theorem %
\ref{thm:ExistenceOfInvMarginal} holds for all possible values of the
parameters. Setting $G(y)=y^{-\theta }$, $y>0,$ in \eqref{eq:Psi_0} yields
\begin{equation*}
\Phi _{0}(y)=(a^{-\theta }-b)c^{-\theta }y^{-\theta }\quad \text{ and }\quad
\Psi _{0}(y)=y^{-\left( \theta +\log _{a}b\right) }.
\end{equation*}%
Since $\theta \neq -\log _{a}b$ and $a\neq 1$, we have the following
dichotomy:

\begin{enumerate}
\item[a)] Either $\theta <-\log _{a}b$ and $a<1$ or $\theta >-\log _{a}b$
and $a>1$. Then, one can show that conditions (i) of Theorem \ref%
{thm:ExistenceOfInvMarginal} hold.

\item[b)] Either $\theta <-\log _{a}b$ and $a>1$ or $\theta >-\log _{a}b$
and $a<1$. Then, one can show that conditions (ii) of Theorem \ref%
{thm:ExistenceOfInvMarginal} hold.\qedhere
\end{enumerate}


\begin{thebibliography}{12}
\providecommand{\natexlab}[1]{#1}
\providecommand{\url}[1]{\texttt{#1}}
\expandafter\ifx\csname urlstyle\endcsname\relax
  \providecommand{\doi}[1]{doi: #1}\else
  \providecommand{\doi}{doi: \begingroup \urlstyle{rm}\Url}\fi


\bibitem[Black(1988)]{Black1988}
F.~Black.
\newblock Individual investment and consumption under uncertainty.
\newblock In D.~L. Luskin, editor, \emph{Portfolio Insurance: A Guide to
 Dynamic Hedging}, pages 207--225. John Wiley and Sons, New York, 1988.
\newblock First version: November 1, 1968, Financial Note No. 6B, Investment
 and consumption through time.


\bibitem[El~Karoui and Mrad(2013)]{ElKarouiMard2013}
N.~El~Karoui and M.~Mrad.
\newblock An exact connection between two solvable {SDEs} and a nonlinear
  utility stochastic {PDE}.
\newblock \emph{SIAM Journal on Financial Mathematics}, 4\penalty0
  (1):\penalty0 697--736, 2013.





\bibitem[Kuczma et~al.(1990)Kuczma, Choczewski, and Ger]{kuczmaEtAl1990}
M.~Kuczma, B.~Choczewski, and R.~Ger.
\newblock \emph{Iterative Functional Equations}.
\newblock Encyclopedia of Mathematics and its Applications. Cambridge
  University Press, 1990.


\bibitem[Musiela and Zariphopoulou(2006)]{MZ2006}
M.~Musiela and T.~Zariphopoulou.
\newblock Investments and forward utilities. Technical report, 2006. URL
  \url{http://www.ma.utexas.edu/users/zariphop/pdfs/TZ-TechnicalReport-4.pdf}.

\bibitem[Musiela and Zariphopoulou(2009)]{MZ2009}
M.~Musiela and T.~Zariphopoulou.
\newblock Portfolio choice under dynamic investment performance criteria.
\newblock \emph{Quantitative Finance}, 9\penalty0 (2):\penalty0 161--170, 2009.

\bibitem[Musiela and Zariphopoulou(2010{\natexlab{a}})]{MZ2010}
M.~Musiela and T.~Zariphopoulou.
\newblock Stochastic partial differential equations in portfolio choice.
\newblock In C.~Chiarella and A.~Novikov, editors, \emph{Contemporary
  Quantitative Finance}, pages 195--216. Springer-Verlag Berlin Heidelberg,
  2010{\natexlab{a}}.

\bibitem[Musiela and Zariphopoulou(2010{\natexlab{b}})]{Musiela-Z-SIFIN}
M.~Musiela and T.~Zariphopoulou.
\newblock Portfolio choice under space-time monotone performance criteria.
\newblock \emph{SIAM Journal on Financial Mathematics}, 1\penalty0
  (1):\penalty0 326--365, 2010{\natexlab{b}}.

\bibitem[Musiela and Zariphopoulou(2011)]{MZ2011}
M.~Musiela and T.~Zariphopoulou.
\newblock Initial investment choice and optimal future allocations under
  time-monotone performance criteria.
\newblock \emph{International Journal of Theoretical and Applied Finance},
  14\penalty0 (01):\penalty0 61--81, 2011.



\bibitem[Nadtochiy and Tehranchi(2017)]{NT2015}
S.~Nadtochiy and M.~Tehranchi.
\newblock Optimal investment for all time horizons and {Martin} boundary of
  space-time diffusions.
\newblock \emph{Mathematical Finance}, 27\penalty0 (2):\penalty0 438--470,
  2017.~

\bibitem[Nadtochiy and Zariphopoulou(2014)]{NZ2014}
S.~Nadtochiy and T.~Zariphopoulou.
\newblock A class of homothetic forward investment performance processes with
  non-zero volatility.
\newblock In Y.~Kabanov, M.~Rutkowski, and T.~Zariphopoulou, editors,
  \emph{Inspired by Finance: The Musiela Festschrift}, pages 475--504.
  Springer, 2014.

\bibitem[Polyanin and Manzhirov(1998)]{PM1998}
A. D. Polyanin and A. V. Manzhirov.
\newblock \emph{Handbook of Integral Equations: Exact Solutions}.
\newblock Faktorial Moscow, 1998.




\bibitem[Shkolnikov et~al.(2016)Shkolnikov, Sircar, and
  Zariphopoulou]{shkolnikovSZ2015}
M.~Shkolnikov, R.~Sircar, and T.~Zariphopoulou.
\newblock Asymptotic analysis of forward performance processes in incomplete
  markets and their ill-posed {HJB} equations.
\newblock \emph{SIAM Journal on Financial Mathematics}, 7\penalty0
  (1):\penalty0 588--618, 2016.

\end{thebibliography}

\end{document}